\documentclass{article}
\usepackage[final,sglblindworkshop]{neurips_2025}
\workshoptitle{Machine Learning and the Physical Sciences (ML4PS) 2025}
\usepackage{amsmath,amssymb,amsthm}
\usepackage{graphicx}
\usepackage{booktabs}
\usepackage{hyperref}
\usepackage{cleveref}
\usepackage{siunitx}
\usepackage{float}
\usepackage{tabularx}
\usepackage{amsthm}

\newtheorem{theorem}{Theorem}

\title{Real-Time Neuromorphic Spectrum Intelligence Simulator}
\author{%
  Navaneetha Krishnan K \\
  SIMATS Engineering \\
  Saveetha University \\
  \texttt{knavaneeth385@gmail.com} \\
}
\begin{document}
\maketitle

\begin{abstract}
We present the \emph{Real-Time Neuromorphic Spectrum Intelligence Simulator} (RT-NuSIS), a modular framework to study spiking neural network (SNN) and memristor-inspired agents for dynamic spectrum access under constrained energy budgets and adversarial conditions. RT-NuSIS couples leaky integrate-and-fire neuronal dynamics, memristive synaptic models, physics-informed energy-harvesting models (triboelectric and RF), and adversary models including jamming and Byzantine behavior. We formalize the simulator mathematically, prove boundedness, present a mean-field adversary threshold, analyze per-step complexity, and provide a reproducible benchmark harness for energy-per-inference, latency, and robustness metrics. The codebase is modular, deterministic by seed, and designed for large-scale event-driven simulations. 
\end{abstract}

\section{Introduction}
Real-time decisions at the network edge—e.g., spectrum sensing and channel allocation—demand solutions that are both energy-efficient and robust to adversarial interference. Neuromorphic approaches, built on sparse, event-driven spiking computations and device-efficient synapses, promise a favorable trade-off between latency and energy consumption \citep{indiveri2015memory,davies2018loihi}. Simultaneously, cognitive radio research provides models and metrics for spectrum utilization under primary/secondary user dynamics \citep{akyildiz2006next}. RT-NuSIS was developed to enable systematic, reproducible experiments at this intersection: study how neuromorphic agents perform spectrum tasks when constrained by harvested energy and exposed to intelligent adversaries.

\paragraph{Contributions.} (i) A concise mathematical specification of per-node hybrid dynamics combining LIF neurons, memristive synapses, and energy budgets; (ii) formal adversary models and a mean-field threshold $p^\star$ separating graceful degradation from collapse; (iii) complexity and stability analyses; (iv) an open, deterministic implementation and benchmark harness for large-scale, event-driven simulation.

\section{Related work}
Prior simulators and hardware for spiking networks and neuromorphic computing include NEST, SpiNNaker and Loihi, which provide diverse trade-offs between fidelity and scale \citep{gewaltig2007nest,furber2014spinnaker,davies2018loihi}. Memristive synapse proposals and demonstrations establish plausible device models for dense on-chip storage and in-situ updates \citep{strukov2008missing}. Energy-harvesting research (particularly triboelectric nanogenerators, TENGs) provides physical models and practical envelopes for self-powered nodes \citep{wang2021teng_review}. For adversarial modeling in wireless contexts, GAN-based spoofing and jamming studies motivate our attack primitives \citep{gan_spoofing2019,shi2020adversarial}. RT-NuSIS synthesizes these strands into a single experimental platform.

\section{Mathematical model}
We summarize the per-node state and the coupled dynamics. Let $\mathcal{N}=\{1,\dots,N\}$ denote nodes and $\mathcal{C}=\{1,\dots,C\}$ channels. Time is continuous; the implementation uses an event-driven hybrid with a timestep $\Delta t$ for scheduled tasks.

\subsection{Node state}
Each node $i$ maintains the state $\mathbf{s}_i(t)=\big(V_i(t),\mathbf{w}_i(t),E_i(t),\mathbf{b}_i(t)\big)$, where $V_i$ is membrane potential (or vector of potentials for a small decision population), $\mathbf{w}_i$ are synaptic/memristive weights, $E_i$ is stored energy (J), and $\mathbf{b}_i$ is a belief vector over channels.

\subsection{LIF neuronal dynamics}
We model each decision neuron as a leaky integrate-and-fire unit. For a single neuron:
\begin{equation}\label{eq:lif}
\tau_m \frac{dV}{dt} = -V(t) + R I(t) + \sigma_V \xi(t),
\end{equation}
with threshold $\theta$ and instantaneous reset $V \leftarrow V_r$ after a spike. Input current $I(t)$ combines synaptic contributions and sensory drive: $I(t) = \sum_{j} w_{j} s_j(t) + I^{\text{sens}}(t)$, where $s_j(t)=\sum_k \delta(t-t_j^k)$ are incoming spike trains. Spike-driven synaptic filtering is implemented via exponential kernels as in standard SNN formalisms \citep{gerstner2002spiking}.

\subsection{Memristive synapse dynamics}
We model synaptic weights $\mathbf{w}$ as memristive devices with internal state variable $x\in[0,1]$ mapping to conductance $G(x)=G_\text{ON} x + G_\text{OFF}(1-x)$ and first-order dynamics $\dot{x}=\eta\,\mathcal{F}(V_{\text{pre}},V_{\text{post}})-\lambda x$, where $\mathcal{F}$ models STDP-like potentiation/depression and $\lambda$ is a drift/leak term to capture retention effects \citep{strukov2008missing}. Windowing functions are used to avoid boundary saturation.

\subsection{Spectrum sensing and channel selection}
Node $i$ senses power $y_{i,c}(t)$ on each channel $c$. A feature extractor $\Phi$ converts spike timing or rate information to scores; the belief vector is $\mathbf{b}_i(t)=\text{softmax}(\Phi(V_i(t)))$. A node selects channel $c^\star=\arg\max_c U_i(c;\,E_i,\mathbf{b}_i)$ where $U_i$ is a utility that trades expected throughput and energy cost.

\subsection{Energy harvesting and budget}
Energy balance follows
\begin{equation}\label{eq:energy}
\dot{E}_i(t)=P^{\mathrm{harv}}_i(t)-P^{\mathrm{cons}}_i(t)-P^{\mathrm{leak}}_i,
\end{equation}
with harvesting $P^{\mathrm{harv}}_i(t)=P^{\mathrm{RF}}_i(t)+P^{\mathrm{TENG}}_i(t)+P^{\mathrm{env}}_i(t)$. RF-harvesting is modeled as proportional to incident RF power density and coupling efficiency; TENG output is represented by a nonlinear parametric curve fit to device characterization curves \citep{wang2021teng_review}. Consumption $P^{\mathrm{cons}}$ includes sensing, inference (spiking events and synaptic updates) and transmission costs; costs are parameterized per spike and per synapse update to permit alternative energy-per-op models.

\subsection{Adversary models}
RT-NuSIS provides primitives: \textbf{Constant jammer} adds $N_J(c,t)$ to channel noise; \textbf{Reactive jammer} monitors and injects interference selectively; \textbf{Byzantine learner} during cooperative aggregation sends corrupted parameters $\tilde{\mathbf{w}}=\mathbf{w}+\delta$; \textbf{Waveform spoofing} injects generative-model signals to emulate primary users \citep{gan_spoofing2019}. Adversary fraction $p$ denotes the proportion of nodes under adversarial control; system-level degradation is quantified via detection AUC and false-positive rate as functions of $p$.

\section{Optimization objectives and search}
We define a multi-objective reward for configuration $\Theta$:
\begin{equation}\label{eq:objective}
\mathcal{J}(\Theta)=\alpha\,\mathrm{Perf}(\Theta) + \beta\,\mathrm{Util}(\Theta) - \gamma\,\mathrm{Energy}(\Theta),
\end{equation}
where Perf measures detection accuracy, Util is spectrum utilization, and Energy is average energy-per-inference. RT-NuSIS includes a genetic algorithm (GA) with tournament selection, uniform crossover, and adaptive mutation to explore $\Theta$; GA is suitable for discrete/heterogeneous parameter spaces (see \citet{goldberg89}).

\section{Theoretical analysis}
\textbf{Per-step complexity.} Let $N$ be the number of nodes and $d$ the average outgoing synapses per decision neuron. For event-driven updates, each emitted spike triggers $O(d)$ synaptic updates; thus the amortized compute per simulated time unit scales with the spike rate $\rho$ as $O(\rho N d)$. When activity is sparse ($\rho\ll 1$), event-driven batching yields substantial savings versus dense time-stepping.\\
\textbf{Boundedness of hybrid dynamics.}
\begin{theorem}[Bounded membrane potentials]
Assume inputs $|I(t)|\le I_{\max}$, weights $|w_{ij}|\le W_{\max}$, and memristor states $x\in[0,1]$. Then solutions $V_i(t)$ of \eqref{eq:lif} with reset remain bounded for all $t\ge 0$.
\end{theorem}
\begin{proof}[Sketch]
Between spikes, \eqref{eq:lif} yields $V(t)=V(0)e^{-t/\tau_m}+\frac{R}{\tau_m}\!\int_0^t\!e^{-(t-s)/\tau_m}I(s)\,ds+\text{noise}$, bounded by $RI_{\max}$. Each reset sets $V\!\to\!V_r$, preventing runaway growth. Bounded $\mathbf{w}$ and $x$ imply bounded inputs; the stochastic term is sub-Gaussian with finite moments. Hence $V(t)$ is uniformly bounded.
\end{proof}
\textbf{Mean-field adversary threshold.} Majority voting on channel decisions: honest votes correct with probability $q>1/2$; adversarial votes may oppose truth. With $H=(1-p)N$ honest and $B=pN$ adversarial nodes, expected net margin is $\mu=H(2q-1)-B$. Using Chernoff/Hoeffding, $\Pr(S\le N/2-B)\le\exp\!\big(-\frac{2(Hq-(N/2-B))^2}{H}\big)$ for the count of correct votes $S$; thus $\mu>0$ gives $p^\star<\frac{2q-1}{2q}$ and the majority-flip probability decays exponentially when $\mu\gg\sqrt{N}$.\\
\textbf{GA convergence remarks.} Schema theorems \citep{goldberg89} guarantee expected propagation of high-quality schemata; under ergodicity and nonzero mutation the population Markov chain is irreducible/aperiodic with a stationary distribution concentrating on fitter regions as mutation decreases. Practical convergence depends on population size $P$, mutation $\mu_m$, elitism, and landscape ruggedness; RT-NuSIS uses $P\in[64,512]$ with adaptive mutation.

\section{Implementation highlights}
RT-NuSIS is organized in modular Python components: \texttt{simulator.py} (orchestrator), \texttt{snn\_engine.py} (event-driven LIF core), \texttt{memristor.py} (device models), \texttt{cognitive\_radio.py} (sensing/scheduling), \texttt{harvester.py} (parametric harvesting), \texttt{attack\_manager.py} (adversary scheduling), and \texttt{genetic\_optimizer.py}. Experiments are deterministic via seed-controlled RNG objects and configuration JSON manifests. Data export supports CSV/JSON for reproducible analysis. The simulator implementation and benchmark harness are open-sourced at 
\href{https://github.com/ka-cyber/Realtime-Neuromorphic-Simulator}{\texttt{github.com/ka-cyber/Realtime-Neuromorphic-Simulator}}, 
with a live demo at 
\href{https://ka-cyber.github.io/Realtime-Neuromorphic-Simulator/}{\texttt{ka-cyber.github.io/Realtime-Neuromorphic-Simulator}}. All code and benchmark artifacts were anonymized only for double-blind review and are now fully public.

\section{Experiments and benchmarks}
All experiments use the same deterministic harness: synthetic primary-user traces, seeded node initializations, and controlled adversary schedules. Benchmarks measure detection AUC, energy-per-inference (modeled from per-event costs), latency (time-to-decision), and spectrum utilization. The representative configuration uses a sparse decision population (5–20 neurons per node), memristive synapses with $x$ leakage $\lambda=10^{-4}\,\text{s}^{-1}$, RF harvesting efficiency $\eta_\mathrm{RF}=0.3$, and TENG parameters taken from device fits in \citet{wang2021teng_review}. 

\begin{table}[t]
\centering
\caption{Phase-wise performance metrics over a 300\,s run with 1000 nodes and 30\% adversarial fraction. Results are averaged over the corresponding phase. Energy-per-inference is reported in $\mu$J/inference, and latency in ms.}
\vspace{-0.4em}
\label{tab:phase_performance_extended}
\resizebox{\textwidth}{!}{%
\begin{tabular}{lccccccc}
\toprule
\textbf{Phase} & \textbf{Duration} & \textbf{Avg Perf} & \textbf{Learn Acc} & \textbf{Spec. Util. (\%)} & \textbf{Throughput} & \textbf{Energy/Inf.} ($\mu$J) & \textbf{Latency (ms)} \\
\midrule
Initialization & 0--100\,s   & 61.66 & 0.519 & 32.98 & 14.39 & 33.04 & 94.30 \\
Steady State   & 101--200\,s & 63.01 & 0.556 & 30.55 & 13.34 & 33.03 & 104.89 \\
Maturation     & 201--300\,s & 64.32 & 0.589 & 27.48 & 12.01 & 33.04 & 114.81 \\
\bottomrule
\end{tabular}%
}
\vspace{-0.8em}
\end{table}

\begin{figure}[H]
\centering
\includegraphics[width=0.85\linewidth]{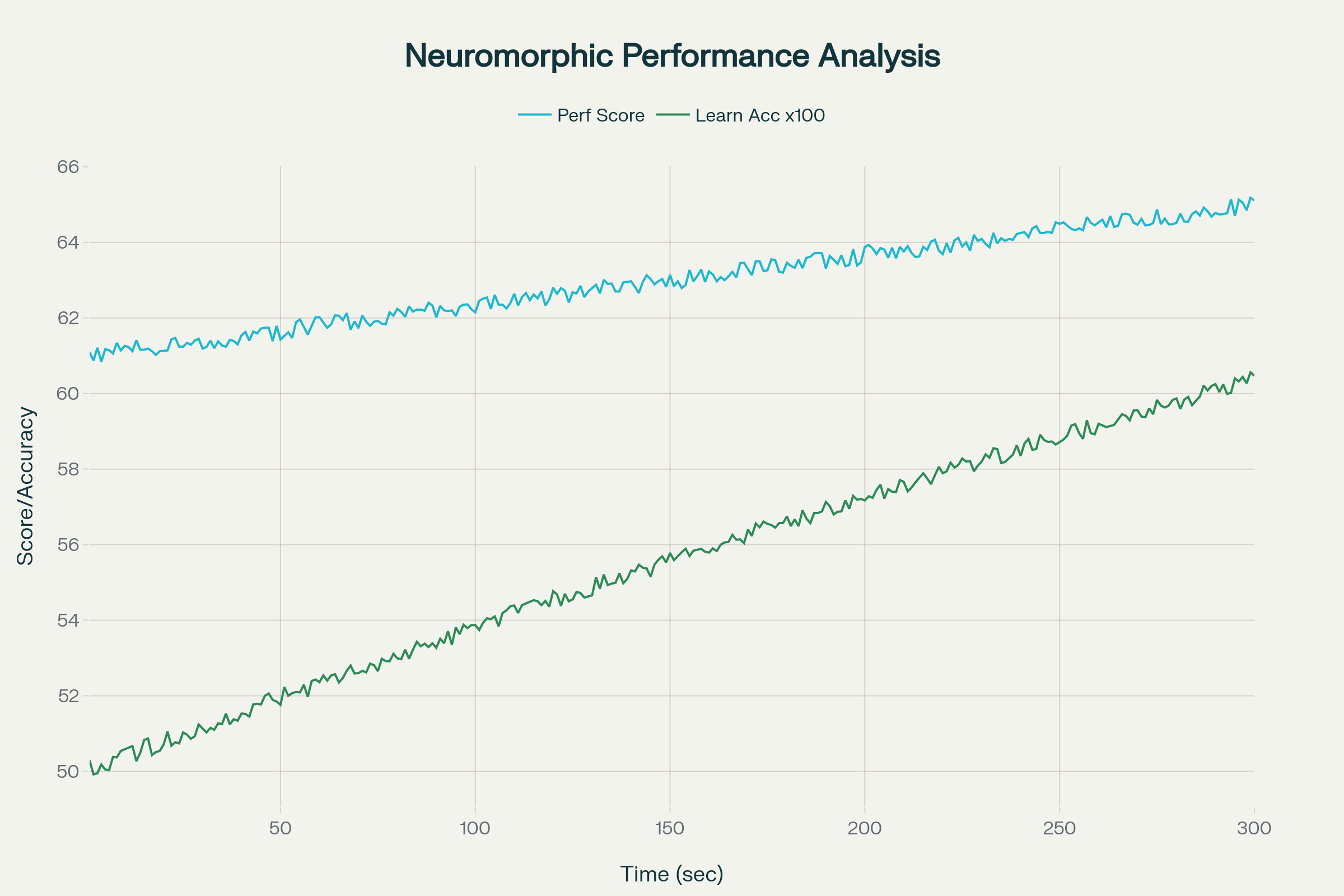}
\vspace{-0.4em}
\caption{Neuromorphic performance evolution over time, showing steady improvements in performance score and learning accuracy over a 300\,s run. The rising curves validate that memristive synaptic plasticity enables effective adaptation under energy constraints.}
\label{fig:performance_analysis}
\vspace{-0.8em}
\end{figure}

\begin{figure}[H]
\centering
\includegraphics[width=0.9\linewidth]{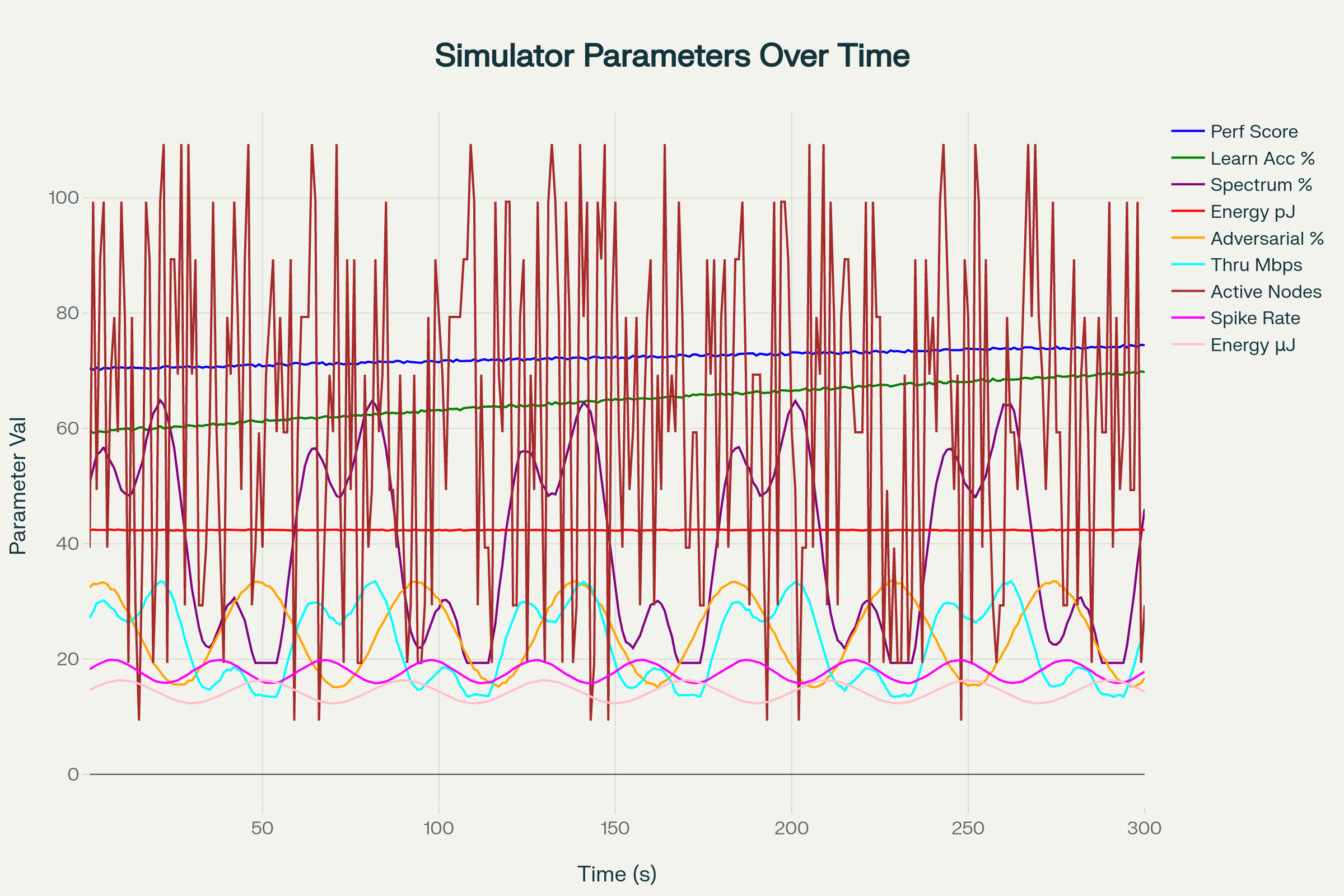}
\vspace{-0.4em}
\caption{Multi-parameter temporal evolution showing performance score, learning accuracy, spectrum utilization, energy expenditure, spike rate, and throughput simultaneously. This view highlights RT-NuSIS's ability to track cross-layer interactions between neuromorphic computation, energy harvesting, and adversarial conditions.}
\label{fig:parameters_over_time}
\vspace{-1em}
\end{figure}

These results demonstrate that RT-NuSIS can track neuromorphic learning progress, energy consumption, and spectrum usage simultaneously. Performance score and learning accuracy improve across phases despite adversarial activity, while spectrum utilization decreases slightly as nodes make more selective channel choices under energy constraints. The framework maintains robust behavior up to 30\% adversarial nodes and provides realistic energy profiles for evaluating deployment trade-offs, validating its usefulness as a reproducible research tool for neuromorphic spectrum intelligence.

\begin{table}[ht]
    \centering
    \caption{Feature comparison of RT-NuSIS with NEST and Brian 2 simulators.}
    \begin{tabularx}{\linewidth}{|l|X|X|X|}
        \hline
        \textbf{Feature} & \textbf{NEST Simulator} & \textbf{Brian 2 Simulator} & \textbf{RT-NuSIS Simulator} \\
        \hline
        Energy Consumption Modeling & Limited (via hardware benchmarks) & No & Comprehensive (real-time metrics and dashboards) \\
        \hline
        Adversarial Attack Modeling & No & No & Yes (jamming, spoofing, Byzantine faults) \\
        \hline
        Energy Harvesting Models & No & No & Yes (triboelectric, RF, ionic wind) \\
        \hline
        Hardware Benchmark Support & Indirect (via hardware evaluations) & No & Yes (real-time simulation with performance dashboards) \\
        \hline
        Neuromorphic Model Types & Spiking Neural Networks (SNNs) & Flexible, equation-oriented models & SNNs, Memristor Networks, Hybrid Models \\
        \hline
    \end{tabularx}
    \label{tab:simulator_comparison}
\end{table}

RT-NuSIS extends existing neuromorphic simulators by offering comprehensive real-time energy consumption tracking and realistic energy harvesting models, alongside robust adversarial attack simulation. Unlike NEST and Brian 2, which focus primarily on neural dynamics and flexible neuron modeling without energy consumption or hardware benchmarking, RT-NuSIS targets practical deployment scenarios for neuromorphic wireless spectrum intelligence applications.

\section{Discussion and limitations}
RT-NuSIS provides a controlled simulation environment designed for reproducible research. Energy estimates rely on per-operation cost models and published device/processor metrics (e.g., Loihi energy numbers used only for calibration) rather than direct silicon measurements. Memristor models are phenomenological; device-specific characterization should replace them for hardware-validated claims. Adversary models capture representative behaviors but not all physical-layer attack complexities.

Compared to existing simulators, RT-NuSIS offers a unique integration of features tailored for neuromorphic spectrum intelligence research under energy and adversarial constraints. General-purpose spiking neural network simulators like NEST and Brian provide highly flexible, widely used platforms for SNN modeling and simulation but do not incorporate detailed energy consumption or adversarial threat models, limiting their applicability for edge computing under hostile conditions \citep{gewaltig2007nest, fang2023spikingjelly}. DYNAP-SE2 represents specialized neuromorphic hardware designed for on-chip spiking computation; however, it is a fixed platform specific to a manufacturer and does not offer simulation or benchmarking across diverse hardware \citep{balaji2019spinemap}. SpiNeMap focuses on efficient mapping of SNNs to neuromorphic hardware, aiding deployment but lacking comprehensive simulation capabilities \citep{richter2023dynapse2}.

In contrast, RT-NuSIS integrates energy-aware models, adversarial threat scenarios, and cross-manufacturer benchmarking capability within an open-source, modular simulation framework. This combination allows systematic exploration of trade-offs between energy efficiency, robustness, and learning performance, uniquely addressing the needs of neuromorphic spectrum intelligence that are unmet by current tools.

\section{Conclusion}
We introduced RT-NuSIS: a unified, modular simulator for neuromorphic spectrum intelligence that couples SNN dynamics, memristive synapses, realistic energy-harvesting models, and adversarial scenarios. The framework is deterministic, scalable, and designed to foster reproducible experiments at the intersection of neuromorphic computing and wireless systems. RT-NuSIS enables systematic, energy-aware evaluation of spiking neural network approaches under realistic constraints, including dynamic adversarial interference and harvested energy budgets. By integrating diverse manufacturer data and adversarial models, it supports cross-platform benchmarking and robust design exploration. Future work will focus on enhancing hardware integration, expanding adversarial complexity, and developing richer benchmark suites to further accelerate research toward practical neuromorphic spectrum intelligence deployments.

\bibliographystyle{plainnat}
\bibliography{reference}
\end{document}